\documentclass[10pt,a4paper]{IEEEtran}

\usepackage{graphicx}
\usepackage{mathptmx}      
\usepackage[hidelinks]{hyperref}
\usepackage{graphicx}
\usepackage{amssymb,amsmath}
\usepackage{breqn}
\usepackage{epstopdf}
\usepackage{cite}
\newtheorem{theorem}{Theorem}[section]
\newtheorem{lemma}[theorem]{Lemma}

\usepackage{subcaption,graphicx}

\newcommand{\forceindent}{\leavevmode{\parindent=2em\indent}}
\newcommand{\supperforceindent}{\leavevmode{\parindent=3em\indent}}

\begin{document}

\title{Analytical and Simulation Performance of a Typical User in Random Cellular Network}
\author{Sinh Cong Lam        \and
Kumbesan Sandrasegaran  \\
Centre for Real-Time Information Networks, Faculty of Engineering and Information Technology, \\
University of Technology, Sydney, Australia
}

\maketitle

\begin{abstract}
Spatial Poisson Point Process (PPP) network, whose Base Stations (BS)s are distributed according to a Poisson distribution, is currently used as a accurate model to analyse the performance of a cellular network. Most current work on evaluation of PPP network in Rayleigh fading channels are usually assumed that the BSs have fixed transmission power levels and there is only a Resource Block (RB) or a user in each cell. In this paper, the Rayleigh-Lognormal fading channels are considered, and it is assumed that each cell is allocated $N$ Resource Blocks (RB) to serve $M$ users. Furthermore, the serving and interfering BS of a typical user are assumed to transmit at different power levels. The closed-form expression for the network coverage probability for both low and high SNR is derived by using Gauss-Legendre approximation.  The analytical results indicates that the performance of the typical user is proportional to the transmission power and density of BSs when $SNR<10$ dB and $\lambda<1$, and reaches the upper bound when $SNR>10$ dB or $\lambda>1$. The variance of Monte Carlo simulation is considered to verify the stability and accuracy of simulation results.

\end{abstract}

\section{Introduction}
In  Orthogonal Frequency-Division Multiple Access (OFDMA) multi-cell networks, the main factor that has a direct impact on  the system performance is InterCell Interference (ICI) which is caused by the use of the same frequency resource  in adjacent cells at the same time. InterCell Interference Coordiantion (ICIC)  techniques have been introduced as an effective technique that can significantly mitigate the ICI and improve users' performance, especially for users experiencing low  Signal-to-Interference-plus-Noise Ratio (SINR).

The two dimensional (2-D) traditional hexagonal network model with deterministic BS locations is the most popular  model that is used to analyze a cellular network. In this model, a service area is divided into several hexagonal cells with \textit{same} radius and each cell is served by a BS which is often located at the center of the cell. Tractable analysis was often achieved for a fixed user with limited number of interfering BSs or in case of ignoring propagation pathloss \cite{Tse}. Another  tractable and simple model is the Wyner model \cite{Wyner} which was developed by information theorists and has been widely used to evaluate the performance of cellular networks in both uplink and downlink directions. In Wyner and its modified models, users were assumed to have fixed locations and interference intensity was assumed to be deterministic and homogeneous. However, for a real wireless network, it is clear that users' locations may be  fixed sometimes, but interference levels vary moderately  depending on several factors such as receiver and transmitter locations, transmission conditions, and the number of instantaneous interfering BSs. 
Hence, these models are no longer accurate to evaluate the performance of multi-cell wireless networks, thus the PPP network model has been proposed and developed as the accurate and flexible tractable model for cellular networks \cite{PPPfirst,Andrews}. 

In PPP model, the service area is partitioned into non-overlapping Voronoi cells \cite{Andrews} in which the number of cells is a random Poisson variable. Each cell is served by a unique BS that is located at its nucleus. Users are distributed as some stationary point process and allowed to connect with the strongest or the closest BSs. In the strongest model, each user measures SINR from several candidate BSs and selects the BS with the highest SINR.  In the closest model, the distances between the user and BSs are estimated, and the BS which is nearest to the user is selected. In this work, we assume that each user associates with the nearest BS.

The PPP network performance can be evaluated by coverage probability approach \cite{Andrews} and Moment Generating Function (MGF) approach \cite{DiRenzo2013}. Coverage probability approach was proposed to calculate the coverage probability and capacity of a typical user that associates with its nearest base station \cite{Andrews}, and then  extended for PPP network enabling frequency reuse \cite{Dhillon2012}. In these work, the closed-form expressions were evaluated  by ignoring Gaussian noise and only in Rayleigh fading.  The closed-form expression for coverage probability is yet  to be investigated and developed for a composite Rayleigh-Lognormal fading channel. MGF approach was proposed in \cite{DiRenzo2013} to avoid the complexity of coverage probability approach. By using this approach, the authors derived the average capacity of a user in a simple PPP network with generalized fading channels. The final equations, however, were not simple because they contained the Gauss hypergeometric function \cite{Stegun1972} which is expressed as an integral.

Some work that evaluated the effects of Rayleigh and shadowing were considered in \cite{Keeler,Xiaobin}. However, in \cite{Keeler}, shadowing was not incorporated in channel gain and assumed to be constant when the origin PPP model is rescaled. Instead of rescaling the network model, authors in \cite{Xiaobin} introduced a new approach to derive the mathematical expression for coverage probability for  PPP network neglecting noise. 

In most of papers, it was assumed that each cell had either a user or a single RB, and all BSs have same power and transmit continuously. These assumptions led to the fact that the neighbouring BSs always created ICI to a typical user. Hence, the impacts of  scheduling algorithms such as Round Robin on network performance were not clearly presented. Furthermore, in all papers that discussed above, the expressions of coverage probability were only presented in the close-form expression in the case of high SNR or neglecting Gaussian noise, otherwise they were presented with two layer integrals which could not be evaluated. 

In this paper, it is assumed that each BS is allocated $N$ RBs to serve $M$ users and has different transmission power. These assumptions are  relevant to the practical network because in cellular networks, the transmission powers of BSs in different tiers such as macro, pico and fermto, are significantly different. Even, the transmission powers of BSs in a given tier still vary and depend on the location or transmission condition. The closed-form expression for coverage probability of a typical user in the closest PPP network model is derived by using coverage probability approach and Gauss-Legendre approximation. A simple part of this paper was presented in \cite{SinhCongLam2015} with assumptions that there is only a RB and a user in the network and all BSs have same transmission power. Furthermore, in this paper, the variance of simulation results is presented to confirm the stable and accuracy of simulation programs.

\section{System model}
Homogeneous Poisson model of wireless network is the simplest PPP  model with a single hierarchical level. In this model, the service area is partitioned into non-overlapping Voronoi cells \cite{PPPfirst,Andrews} in which the number of cells is a random Poisson variable. Each cell is served by a unique BS that is located at its nucleus (see Figure \ref{fig:PPPnetwork}). Users are distributed as some stationary point process and allowed to connect with the closest BSs. 

\begin{figure}[h!]
\includegraphics[width=9cm,height=6.5cm]{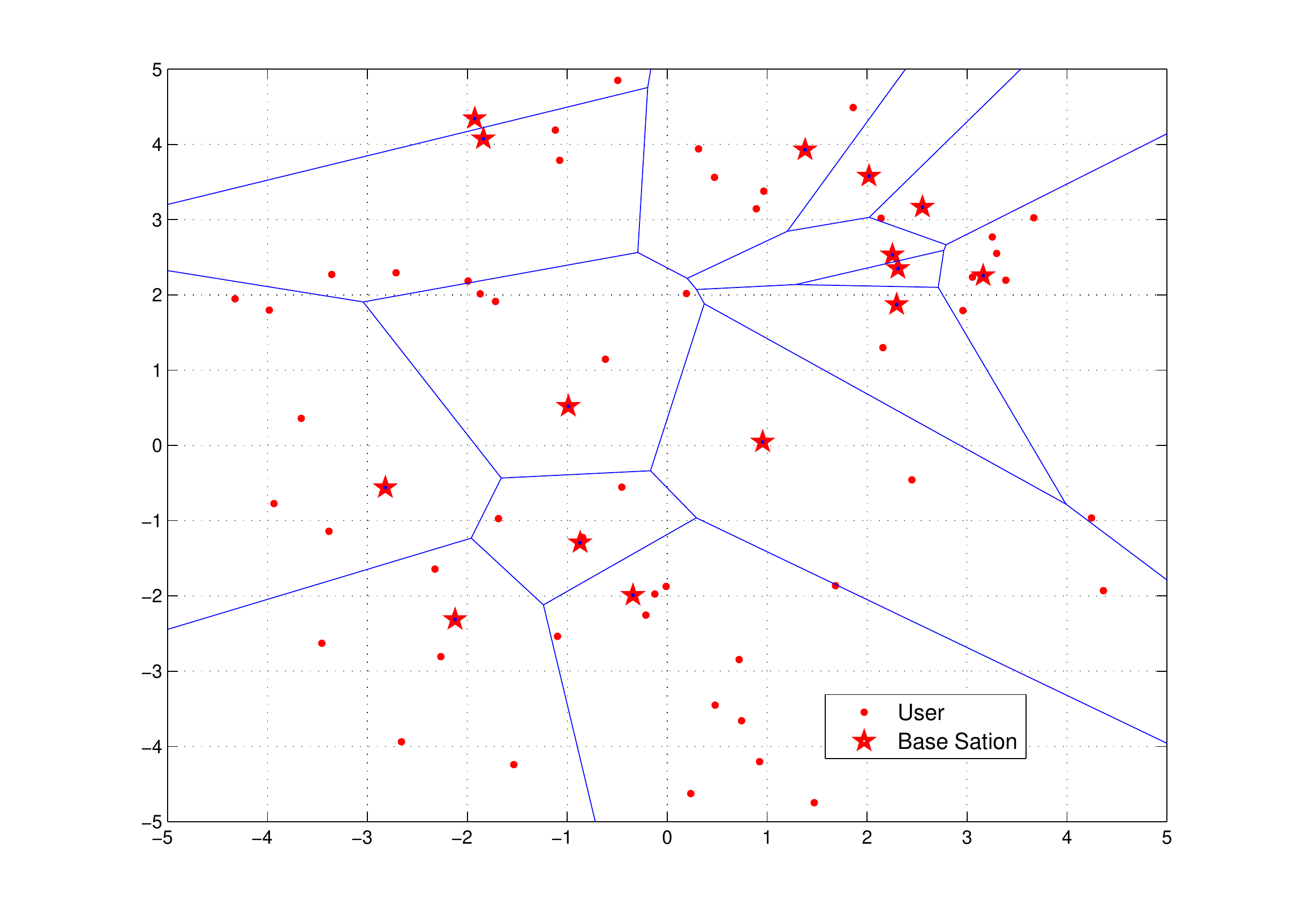}
\caption{An example of a network model in which the density of BSs and users are $\lambda=0.25$ and $\lambda=0.75$}
\label{fig:PPPnetwork}
\end{figure}

In the nearest model, an importance parameter $r$ is defined as the distance from a typical user to its associated BS. Since each user connects with the closest BS, all neighboring BSs must be further than $r$. The null probability of a 2-D Poisson process with density $\lambda$ in a globular area with radius $R$ is $\exp(-2\pi\lambda R^2)$,  then the  Cumulative Distribution Function (CDF)  of $r$ is given by  \cite{Andrews,ThomasDavidNovlan2011}:.
\begin{equation*}
F_R(R)=\mathbb{P}(r<R)=1-\mathbb{P}(r>R)=1-e^{-\pi \lambda R^2}
\end{equation*}
The  PDF  can be obtained by finding the derivative of the CDF: 
\begin{equation}
f_R(r)=\frac{dF_r(r)}{dr}=2\pi\lambda re^{-\lambda \pi r^2}
\label{PDFdistance}
\end{equation}

In Figure \ref{fig:PPPnetwork}, a 6 km x 6 km service area is considered where the distribution of BSs is a Poisson Spatial Process with density $\lambda=0.25$. It can observed that the boundaries of the cell as well as the locations of BSs in this model are generated randomly to correspond  with the changes of network operations. The main weakness of this model is that sometimes BSs  are located very close together, but this can be overcome by taking the average from multiple results of network performance. 

In this paper, it is assumed that every cell in the network has $M$ users and is allocated $N$ resource block (RB). The probability where the probability where a BS causes Intercell Interference (ICI) to a typical user is represented by a indicator function $\tau(RB_i=RB_j)$. This indicator function takes values 1 if the base station in cell $i$ and $j$ transmit on the same RB at the same time.  When the Round Robin scheduling is deployed, the expected values of $\tau(RB_i=RB_j)$ is archived by:
 \begin{equation*}
\mathbb{E}(\tau(RB_i=RB_j))=\frac{M}{N}=\epsilon
 \end{equation*}

\subsection{Downlink network model}
In downlink cellular network, the  transmitted signal from a BS usually experiences  multiple propagation phenomena including fast fading, slow fading and path loss \cite{pathloss}. Fast fading is caused by multipath propagation phenomena that results in rapid fluctuations of the received signal in terms of phase and amplitude. Slow fading, which occurs as the signal travels through  large obstructions such as  buildings or hills, leads to the slower phase and amplitude changes over the period of transmission. Path loss is a natural phenomenon in which the transmitted signal power gradually reduces when it travels over a distance. In this session, we will discuss about the statistical models of these propagation phenomena.

\subsection*{Statistical path loss model}
\label{section:SINR}
In most statistical models of wireless networks, it is assumed that all receiver antennas have the same gain and height.  The received signal power at a receiver at a distance $r$ from the transmitter can be given by Equation \ref{eq:pathloss} \cite{pathloss}:
\begin{align}
P_r=\zeta Pr^{-\alpha}
\label{eq:pathloss}
\end{align}
 The propagation path loss in dB unit is obtained by
\begin{equation}
PL (dB) = 10\log_{10}\left(\frac{P_r}{\zeta P} \right) =-10\alpha\log_{10}r
\label{eq:pathlossdB}
\end{equation}
 in which $\alpha$ is path loss exponent; P and $\zeta$ are standard transmission power of a BS and a power adjustment coefficient, respectively, $\zeta>0$. The values of $\alpha$, which were found from field measurements  are listed in Table \ref{table:pathloss}\cite{pathlossexponent}
 
\begin{table}[!h]
\centering
\small
\begin{tabular}{|c|c|}
\hline \rule[0ex]{0pt}{2.5ex} Environment & Path loss coefficient \\
\hline \rule[0ex]{0pt}{2.5ex} Free space  & 2 \\
\hline \rule[0ex]{0pt}{2.5ex} Urban Area  & 2.7 - 3.5 \\
\hline \rule[0ex]{0pt}{2.5ex} Suburban Area & 3 - 5 \\
\hline \rule[0ex]{0pt}{2.5ex} Indoor (line-of-sight) & 1.6 - 1.8\\
\hline 
\end{tabular}
\caption{Propagation path loss coefficient}
\label{table:pathloss}
\end{table}

Due to the variation of $\alpha$  with changes of transmission environment, as a signal propagates over a wide range of areas, it can be  affected  by different attenuation mechanisms. For example, the first propagation area near the BS is  free-space  area where $\alpha=2$ and the second area  closer to the user may be heavily-attenuated area such as urban area where $\alpha=3$. In a real network, the path loss can be estimated by  measuring signal strength and then be overcome by increasing the transmission power.

\subsection*{Fading channel model}
The multipath effect at the mobile receiver due to scattering from local scatters such as buildings in the neighborhood of the receiver causes  fast fading, while the variation in the terrain configuration between the base-station and the mobile receiver causes  slow shadowing (Figure \ref{fig:fading}) .

\begin{figure}[tbph]
\includegraphics[width=1\linewidth]{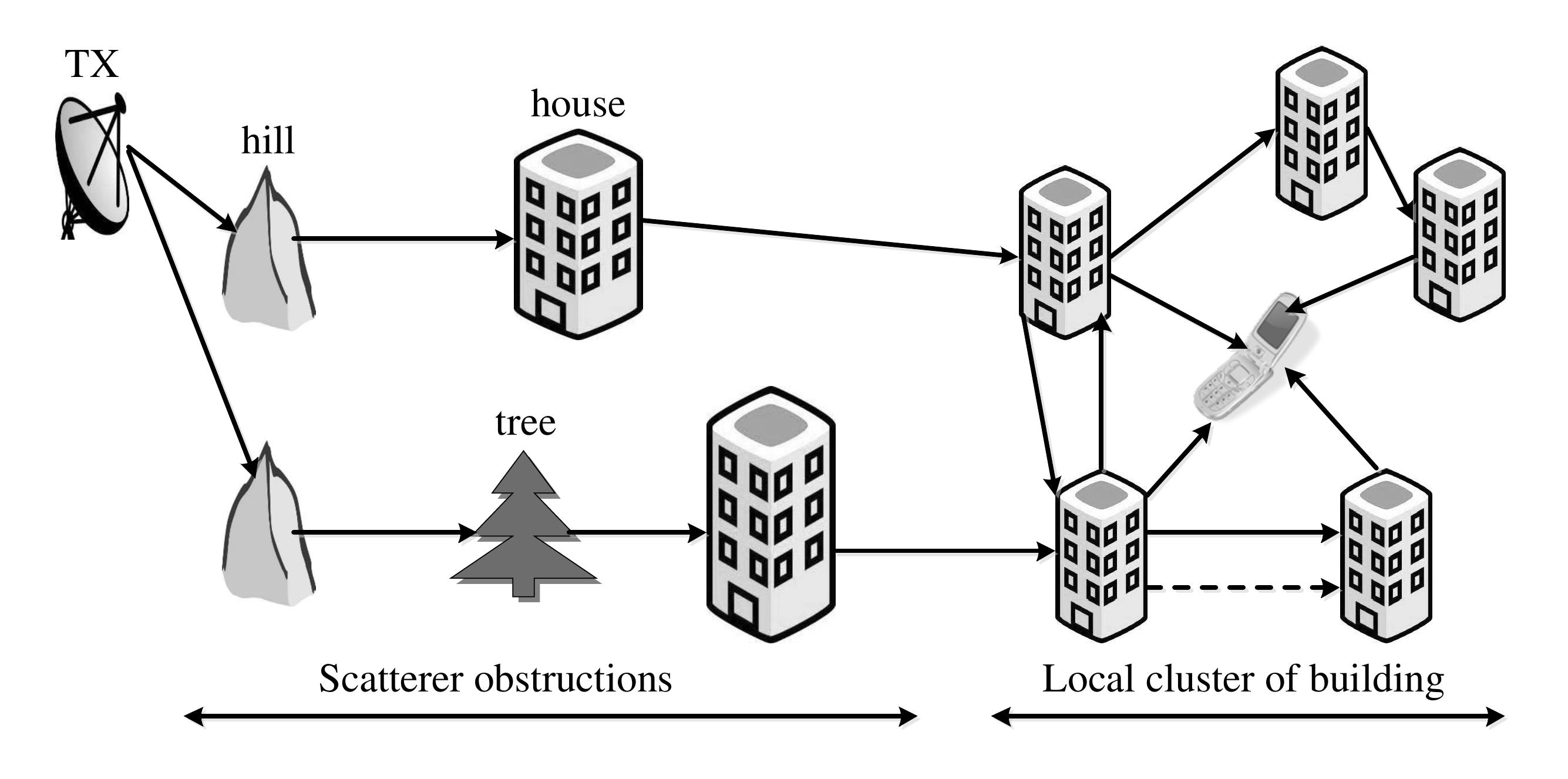}
\caption{ Typical mobile radio propagation topology}
\label{fig:fading}
\end{figure}

 The received signal envelope is  composed of a small scale multipath fading component superimposed on a larger scale or slower shadowing component.  The signal envelope of the multipath component can be modeled as a Rayleigh distributed RV, and its power can be modeled as an exponential RV. Thus, the path power gain has a mixed Rayleigh-Lognormal distribution which is also known as the Suzuki fading distribution model \cite{DinhThiThai}.

The PDF of power gain $g$ of a signal experiencing Rayleigh and Lognormal fading is found from the PDF of the product two cascade channels \cite{DinhThiThai}.
\begin{align}
f&_{R-Ln}(g)=\nonumber \\ &\int_{0}^{\infty}\frac{1}{x}\exp\left(-\frac{g}{x}\right)\frac{1}{x\sigma_z\sqrt{2\pi}}\exp\left(-\frac{\left(10\log_{10}x-\mu_z \right)^2 }{2\sigma_z^2} \right)dx   
\label{eq:originalCompo}
\end{align}
in which $\mu_z$ and $\sigma_z$ are mean and variance of Rayleigh-Lognormal random variable.\\
Using the substitution,  $t=\frac{10\log_{10}x-\mu_z}{\sqrt{2}\sigma_z}$, then
\begin{align*}
x&=10^{(\sqrt{2}\sigma_zt+\mu_z)/10} \triangleq \gamma(t) \\ \text{and} \quad
dx&=\sqrt{2}\sigma_z 10^{(\sqrt{2}\sigma_zt+\mu_z)/10}dt = x\sqrt{2}\sigma_z dt
\end{align*} 
The Equation \ref{eq:originalCompo} becomes
\begin{align}
f_{R-Ln}(g)&=\int_{-\infty}^{\infty}\frac{1}{\sqrt{\pi}}\frac{1}{\gamma(t)}\exp\left(-\frac{g}{\gamma(t)} \right) \exp(-t^2)dt
\label{tmpGauss-Hermite}
\end{align}
The integral in Equation \ref{tmpGauss-Hermite} has the suitable form for Gauss-Hermite expansion approximation \cite{Stegun1972}. Thus, the  PDF  can be approximated by:
\begin{equation}
f_{R-Ln}(g)=\sum_{n=1}^{N_p}\frac{\omega_n}{\sqrt{\pi}}\frac{1}{\gamma(a_n)}\exp\left(-\frac{g}{\gamma(a_n)} \right) 
\label{eq:Rayleigh-Lognormalpdf}
\end{equation}
in which
\begin{itemize}
	\item 	$w_n$  and $a_n$ are the weights and the abscissas of the Gauss-Hermite polynomial respectively. The approximation becomes more accurate with increasing approximation order $Np$. For sufficient approximation, $Np=12$ is used.
	\item  $\gamma(a_n)=10^{(\sqrt{2}\sigma_za_n+\mu_z)/10}$.
\end{itemize}

Hence, the  CDF  of Rayleigh-Lognormal  RV  $F_{R-Ln}(g)$  is obtained by the integral of  PDF  from 0 to $g$, and is derived in the  following steps:
\begin{align}
F_{R-Ln}(g)=&\int\limits_{0}^{g}f(x)dx \nonumber \\
=&\int\limits_{0}^{g}\sum_{n=1}^{N_p}\frac{\omega_n}{\sqrt{\pi}}\frac{1}{\gamma(a_n)}\exp\left(-\frac{x}{\gamma(a_n)} \right)dx \nonumber \\
=&\sum_{n=1}^{N_p}\frac{\omega_n}{\sqrt{\pi}}\frac{1}{\gamma(a_n)}\int\limits_{0}^{g}\exp\left(-\frac{x}{\gamma(a_n)} \right)dx \nonumber \\
=&\sum_{n=1}^{N_p}\frac{\omega_n}{\sqrt{\pi}}\left(1-\exp\left(-\frac{g}{\gamma(a_n)} \right)  \right) 
\end{align}

Since $g$ is defined as the channel power gain, $g$ is a positive real number $(g>0)$. The MGF of $g$ can be found as shown below:
\begin{align}
M_{R-Ln}(s)&=\int\limits_{0}^{\infty}f_{R-Ln}(x)e^{-xs}dx \nonumber \\
&=\int\limits_{0}^{\infty}\sum_{n=1}^{N_p}\frac{\omega_n}{\sqrt{\pi}}\frac{1}{\gamma(a_n)}\exp\left(-\frac{x}{\gamma(a_n)} \right) e^{-xs}dx \nonumber\\
&=\sum_{n=1}^{N_p}\frac{\omega_n}{\sqrt{\pi}}\frac{1}{\gamma(a_n)}\int\limits_{0}^{\infty}\exp\left[-x\left(\frac{1}{\gamma(a_n)}+s\right) \right] dx \nonumber \\
&=\sum_{n=1}^{N_p}\frac{\omega_n}{\sqrt{\pi}}\frac{1}{1+s\gamma(a_n)}
\label{eq:MGF}
\end{align}

\subsection*{Signal-to-Interference-Noise (SINR)}
\label{section:SINR}
The received  signal power for  a  user that is communicating with  it’s serving BS at a distance $r$  and a channel power gain $g$ is given by :
\begin{equation}
S(r)=\zeta Pgr^{-\alpha}
\label{eq:desiredsignalpower}
\end{equation}
The set of interfering BSs is denoted as $\theta$; $r_u$ and 
$g_u$ are the distance and channel power gain from a user to an interfering BS, respectively. The interfering BSs are assumed to transmit at the same power $P_u=\rho P (\rho>0)$. The intercell interference at a user is obtained by
\begin{equation}
I_\theta=\sum_{u\in\theta}\tau(RB_i=RB_j)P_ug_ur_u^{-\alpha}=\sum_{u\in\theta}\rho \tau(RB_i=RB_j)Pg_ur_u^{-\alpha}
 \label{eq:interfernce}
\end{equation}

Combining Equation \ref{eq:desiredsignalpower} and \ref{eq:interfernce}, the received instantaneous SINR(r) at a  user is found from Equation \ref{eq:receivedSINR}
\begin{equation}
SINR(r)=\frac{\zeta Pgr^{-\alpha}}{\sigma^2+I_\theta}
\label{eq:receivedSINR}
\end{equation}
where $\sigma^2$ denotes the Gaussian noise at the receiver.

\section{Coverage probability}
 The coverage probability $P_c$ of a typical user at a distance $r$ from its serving BS for a given $SINR(r)$ threshold $T_c$ is defined as the probability of event in which the received SINR in  Equation \ref{eq:receivedSINR} is larger than a threshold. In other words, if the received SINR(r) at a user is larger than SINR threshold $T_c$, the user can successfully decode the received signal and communicate with the serving BS.   The value of $T_c$ is  dependent on the receiver sensitivity of the UE. The coverage probability $P_c$ can be written as a function of SINR threshold  $T_c$, BS density $\lambda$ and attenuation coefficient $\alpha$ and the distance between the user and its serving BS:
 \begin{equation}
 \mathbb{P}_c(T_c,\lambda,\alpha,r)=\mathbb{P}(SINR(r)>T_c)
 \label{coveragadefinition}
 \end{equation}
or
\begin{equation}
\mathbb{P}_c(T_c,\lambda,\alpha,r)=\mathbb{P}\left(\frac{\zeta Pgr^{-\alpha}}{\sum_{u\in\theta}\tau(RB_i=RB_j)P_ug_ur_u^{-\alpha}+\sigma^2} > T_c \right)
\label{eq:coverageSINRform}
\end{equation}

For a given user, if $r$ is the distance from the user to its serving BS  then $SINR(r)$ depends on the power gain from BS $g$,  the power gain from interfering BS   $g_u$,  $\theta$ is the set of interfering BS, and $r_u$ is the distance  from a user to its interfering BS.  In Equation \ref{eq:coverageSINRform}, $\mathbb{P}$ stands for the conditional average coverage probability and it is expressed as a function of variables $g, g_u, r_u$ and $\theta$, then Equation \ref{eq:coverageSINRform} can be written as 
\begin{align}
\mathbb{P}_c&(T_c,\lambda,\alpha,r)=\nonumber \\
&\mathbb{P}_{\text{\textit{g}},\text{\textit{g}}_{\text{\textit{u}}},\text{\textit{r}}_{\text{\textit{u}}},\theta}\left(\frac{\zeta Pgr^{-\alpha}}{\sum_{u\in\theta}\tau(RB_i=RB_j)P_ug_ur_u^{-\alpha}+\sigma^2} > T_c \right)
\label{eq:coveragedefinitionFinal}
\end{align}

\begin{theorem}
The coverage probability of a typical user in Rayleigh-Lognormal fading in which BSs are distributed as PPP with density $\lambda$ and are allocated $N$ sub-bands randomly is given by 
\begin{equation}
	P_c (T_c,\lambda,\alpha,r)=\sum_{n=1}^{N_p}\frac{w_n}{\sqrt{\pi}}e^{-\frac{T_c}{\gamma(a_n)}\frac{1}{\zeta SNR}r^\alpha } e^{-\pi\lambda\epsilon r^2 f_{I} (T_c,n)}
	\label{eq:singlecoveragePro}
\end{equation}
where $SNR=\frac{P}{\sigma^2}$ is the signal-to-noise ratio at the transmitter, $C=T_c\frac{\rho}{\zeta}\frac{\gamma(a_{n1})}{\gamma(a_n)}$; $f_{I}(T_c,n)$ is defined in Equation \ref{eq:singlefinalexp}. 
\label{theo:singlecoverage}
\end{theorem}

\begin{proof}:
See the Appendix.
\end{proof}

It is observed that there are two exponential parts in Equation \ref{eq:singlecoveragePro}. The first part, i.e $e^{-\frac{T_c}{\gamma(a_n)}\frac{1}{\zeta SNR}r^\alpha}$, which represents the transmission power of the serving BS $\zeta SNR$ and the coverage threshold $T_c$, indicates that the coverage probability is proportional to $\zeta SNR$ . The second part, i.e $e^{-\pi\lambda\epsilon r^2 f_{I} (T_c,n)}$, which represents the ICI, indicates that the coverage probability is inversely proportional to the exponential function of the ratio between the number of users and RBs.

\begin{lemma}
The average coverage probability of a typical user over a cellular network with composite Rayleigh-Lognormal fading is 
\begin{align}
\overline{P}_c(T,\lambda,\alpha)&=4\pi\lambda\sum_{m=1}^{N_{GL}}\frac{c_m(x_m+1)}{(1-x_m)^3}e^{-\pi\lambda\left(\frac{x_m+1}{1-x_m} \right)^2}\nonumber\\
&\sum_{n=1}^{N_p}\frac{w_n}{\sqrt{\pi}}e^{-\frac{T_c}{\gamma(a_n)}\frac{1}{\zeta SNR}\left(\frac{x_m+1}{1-x_m}\right) ^\alpha } e^{-\pi\lambda\epsilon \left(\frac{x_m+1}{1-x_m}\right) ^2 f_{I} (T_c,n)}
\end{align}
\label{lemmaAverageCov}
\end{lemma}
in which $c_m$ and $x_m$ are weights and nodes of Gauss-Legendre rule with order $N_{GL}$; $P_c$ as defined in Equation \ref{theo:singlecoverage}.

\begin{proof}
The average coverage probability is achieved by taking the expected value of $P_c(T_c,\lambda,\alpha,r)$ in Equation \ref{eq:singlecoveragePro} with variable $r>0$ 
\begin{align}
\overline{P}_c(T_c,\lambda,\alpha)&=\mathbb{E}\left(P_c(T_c,\lambda,\alpha,r)\right)\nonumber\\ 
&=\int_{0}^{\infty}P_c(T_c,\lambda,\alpha,r)f_R(r)dr\nonumber \\
&=\int_{0}^{\infty}2\pi\lambda re^{-\pi\lambda r^2}P_c(T_c,\lambda,\alpha,r)dr
\label{eq:singlecoverageProintegral}
\end{align}

\begin{equation*}
\hspace{-2.5cm}
\text{Let} \quad  r=\frac{t}{1-t}  \Rightarrow 
  \begin{cases}
    0 < t < 1; \\
    t=\frac{r}{r+1};\\
    dx=\frac{1}{(1-t)^2};
  \end{cases}
  \text{then,}
\end{equation*}

\begin{align}
\overline{P}_c(T_c,\lambda,\alpha)&=2\pi\lambda \int\limits_{0}^{1}\frac{t}{(1-t)^3}e^{-\pi\lambda(\frac{t}{1-t})^2}P_c(T_c,\lambda,\alpha,\frac{t}{1-t})dt \nonumber
\end{align}
\begin{equation*}
\hspace{-2.5cm}
\text{Let} \quad t=\frac{1}{2}z+\frac{1}{2}  \Rightarrow 
  \begin{cases}
    -1<z<1; \\
    z=2t-1;\\
    dt=\frac{1}{2}dt;
  \end{cases}
  \text{then,}
\end{equation*}
\begin{align}
\overline{P}_c(T_c,\lambda,&\alpha)=\nonumber \\
&4\pi\lambda \int\limits_{-1}^{1}\frac{z+1}{(1-z)^3}e^{-\pi\lambda(\frac{z+1}{1-z})^2}P_c(T_c,\lambda,\alpha,\frac{z+1}{1-z})dz
\label{eq:tmpGaussLegenre}
\end{align}
The integral in Equation \ref{eq:tmpGaussLegenre} has the suitable form of Gauss-Legendre approximation. Hence, the average coverage probability is approximated by
\begin{align}
\overline{P}_c(T,\lambda,\alpha)&=4\pi\lambda\sum_{m=1}^{N_{GL}}\frac{c_m(x_m+1)}{(1-x_m)^3}e^{-\pi\lambda\left(\frac{x_m+1}{1-x_m} \right)^2}\nonumber\\
&\sum_{n=1}^{N_p}\frac{w_n}{\sqrt{\pi}}e^{-\frac{T_c}{\gamma(a_n)}\frac{1}{\zeta SNR}\left(\frac{x_m+1}{1-x_m}\right) ^\alpha } e^{-\pi\lambda\epsilon \left(\frac{x_m+1}{1-x_m}\right) ^2 f_{I} (T_c,n)}
\label{eq:averagecoverP}
\end{align}
The Lemma \ref{lemmaAverageCov} is proved.
\end{proof}

The close-form expression of the average coverage probability has been not yet been derived. Hence, the use of Gauss-Legendre rules is considered as the appropriate approach to find the close-form expression.

For $\sigma^2=0$ or high $SNR$, the average coverage probability can be achieved as follows: 
\begin{align}
\overline{P}_c(T,\alpha) &= \int_{0}^{\infty}2\pi\lambda re^{-\pi\lambda r^2}P_c(T,\lambda,\alpha,r)dr \nonumber \\
&=\int_{0}^{\infty}2\pi\lambda re^{-\pi\lambda r^2}\sum_{n=1}^{N_p}\frac{w_n}{\sqrt{\pi}}e^{-\pi\lambda\epsilon r^2 f_{I} (T,n)}dr \nonumber \\
&=\sum_{n=1}^{N_p}\frac{w_n}{\sqrt{\pi}}\int_{0}^{\infty}2\pi\lambda re^{-\pi\lambda r^2\left(1+\epsilon f_{I}(T,n)\right)}dr \nonumber \\
&=\sum_{n=1}^{N_p}\frac{\omega_n}{\sqrt{\pi}}\frac{1}{1+\epsilon f_{I}(T,n)}
\label{eq:singleCoverHighSNR}
\end{align}
This is the close-form expression of the average coverage probability of a typical user in the interference-limited PPP network. It is observed from equation that the average coverage probability does not depend on the density of BS which means the power of the desired signal in this case counter-balanced with the power of ICI. This results is comparable with others that were published in \cite{Andrews,Dhillon2012} for the case of Rayleigh fading and a single user. 
\begin{lemma}
The coverage probability of a typical user over network in Rayleigh fading only.
\begin{equation}
P_c(T,\lambda,\alpha,r)=e^{-T\frac{1}{\zeta SNR}r^\alpha}e^{-\pi\lambda\epsilon r^2 f_{I}(T,N,n)}
\label{eq:singleCovRay}
\end{equation}
where 
\begin{dmath}
f_{I}(T)=\frac{2}{\alpha}C^{\frac{2}{\alpha}}\frac{\pi}{\sin\left(\frac{\pi(\alpha-2)}{\alpha} \right) } +  \sum_{m=1}^{N_{GL}}\frac{c_m}{2}\frac{C}{C+\left(\frac{x_m+1}{2} \right)^{\alpha/2}}
\end{dmath}
where $SNR=\frac{P}{\sigma^2}$ is the signal-to-noise ratio at the transmitter, $C=T_c\frac{\rho}{\zeta}\frac{\gamma(a_{n1})}{\gamma(a_n)}$; $f_{I}(T_c,n)$ is defined in Equation \ref{eq:singlefinalexp}. 
\end{lemma}

\begin{proof}
Rayleigh fading is a special case of composite Rayleigh-Lognormal fading with $\sigma_z=0 $ and given that $\sum_{n=1}^{N_p}\frac{\omega_n}{\sqrt{\pi}}=1$, then the coverage probability in this case is derived by Equation \ref{eq:singleCovRay}.
\end{proof}

The average coverage probability over network is calculated by integrating Equation \ref{eq:singlecoveragePro} with variable $r>0$, and then its closed-form is expressed as in Equation \ref{eq:singlecoveragePro} where $P_c(T,\lambda,\alpha,r)$ was defined in Equation \ref{eq:singleCovRay}. This analytical result is comparable to the corresponding result for Rayleigh fading given in \cite{Andrews}.

\section{Average capacity}
The average rate, i.e. ergodic rate, of a typical randomly user located in the network is defined as
\begin{equation}
R= \mathbb{E}_t\left[\ln(SINR(r)+1) \right] 
\end{equation}
where $SINR(r)$ is the received SINR at the user given in Equation \ref{eq:receivedSINR}; $\mathbb{E}_t$ represents the conditional expected values of $\ln(SINR(r)+1)$ over the PPP network with variable $t=T_c$. Since $\mathbb{E}(X)=\int\limits_{t>0}\mathbb{P}(X>t),\quad \forall X>0$,
\begin{align}
R&=\int\limits_{0}^{\infty}\mathbb{P}\left[\ln(SINR(r)+1)>t\right]dt \nonumber \\
&=\int\limits_{0}^{\infty}\mathbb{P}\left[SINR(r)>e^t-1\right]dt \nonumber \\
&=\int\limits_{0}^{\infty}\overline{P}_c(e^t-1,\lambda,\alpha)dt
\end{align}
in which $\overline{P}_c(e^t-1,\lambda,\alpha)$ is the average coverage probability of the typical user in the PPP network and obtained by Equation \ref{eq:averagecoverP}.

Using the similar approach in Theorem \ref{lemmaAverageCov}, the average rate can be  approximated by
\begin{equation}
\label{eq:singleCapacity}
R=\sum_{1i=1}^{N_{GL}}\frac{2c_{1i}}{(1-x_{1i})^2}\overline{P}_c(z(x_{1i}),\lambda,\alpha)
\end{equation} 
where $c_{1i}$ and $x_{1i}$ are weights and nodes of Gauss-Legendre rule with order $N_{GL}$; $z(x_{1i})=\exp\left(\frac{x_{1i}+1}{1-x_{1i}} \right)-1 $ and $P_c$ is defined in Equation \ref{eq:averagecoverP}.

\section{Simulation and discussion}
\subsection{Simulation setup}
The simulation algorithms is described in the following steps:

\noindent\rule{9cm}{0.9pt}
\textit{\textbf{for i=1:1:NoR}} \newline
	\forceindent count = 0;\\
	\forceindent \textit{\textbf{for i=1:1:NoS}}\\
		 \supperforceindent\textit{1. Generate $N$ numbers of BSs}\\
		 \supperforceindent\textit{2. Generate $N$ distances between a user and BSs.} \\
		 \supperforceindent\textit{3. Generate $N$ Rayleigh-Lognormal power gain values.}\\
		 \supperforceindent\textit{4. Calculate SINR.}\\
		 \supperforceindent\textit{5. Count outage event}\\
		 \supperforceindent \forceindent	if $SINR< threshold$\\
		 \supperforceindent \supperforceindent	\textit{count=count+1;}\\
		 \supperforceindent \forceindent	end\\
	\forceindent\textbf{end}\\
	\forceindent Coverage Probability \textit{P=count/NoS;}\\
\textbf{end}\\
Variance is obtained by Equation \ref{eq:variance}\\
\noindent\rule{9cm}{0.9pt}
in which $NoR$ and $NoS$ are number of simulation runs and samples per each run, respectively. Higher values of $NoR$ and $NoS$ give more accurate and stable results, however, it takes time and requires high performance computers. In this work, $NoR=5$ and $NoS=10^5$ are appropriate choices to obtain the acceptable variance of simulation results (smaller than $0.001$).

\subsection{Simulation results}
The relationship between coverage probability and related parameters are validated and visualized by Monte Carlo simulations as shown in the following figures. The simulation parameters in figures (\textit{if be not mentioned in figures}) are summarised in Table \ref{SimulationPara}.
\begin{table}[!h]
\centering
\begin{tabular}{|l|l|}
\hline \rule[0ex]{0pt}{3ex} \textbf{Parameter} & \textbf{Value} \\ 
\hline \rule[0ex]{0pt}{2ex} Density of BSs & $\lambda=0.25$ \\ 
\hline \rule[0ex]{0pt}{2ex} Number of RBs & 15 \\ 
\hline \rule[0ex]{0pt}{2ex} Standard transmission power & $SNR=10$ (dB) \\ 
\hline \rule[0ex]{0pt}{2ex} Power adjustment coefficient  & $\zeta=1$ \\ 
							of serving BS  &  \\
\hline \rule[0ex]{0pt}{2ex} Coverage threshold & $T_c=0$ (dB) \\
\hline \rule[0ex]{0pt}{2ex} Fading channel & $\mu_z=-7.3683$ dB \\ 
		\rule[0ex]{0pt}{2ex}& $\sigma_z=8$ dB \\
\hline \rule[0ex]{0pt}{2ex} Pathloss exponent & $\alpha=3.5$ \\ 
\hline 
\end{tabular} 
\caption{Analytical and simulation parameters}
\label{SimulationPara}
\end{table} 

With higher values of $\alpha$, total power of interfering signals decreases at  a faster  rate with distance compared to  desired signal since the user receives only one useful signal from serving cell and often suffers more than one interfering signals. The average coverage probability is, hence, inversely proportional to path loss exponent $\alpha$. 
\begin{figure}[!h]
\centering
\includegraphics[width=9cm,height=6.5cm]{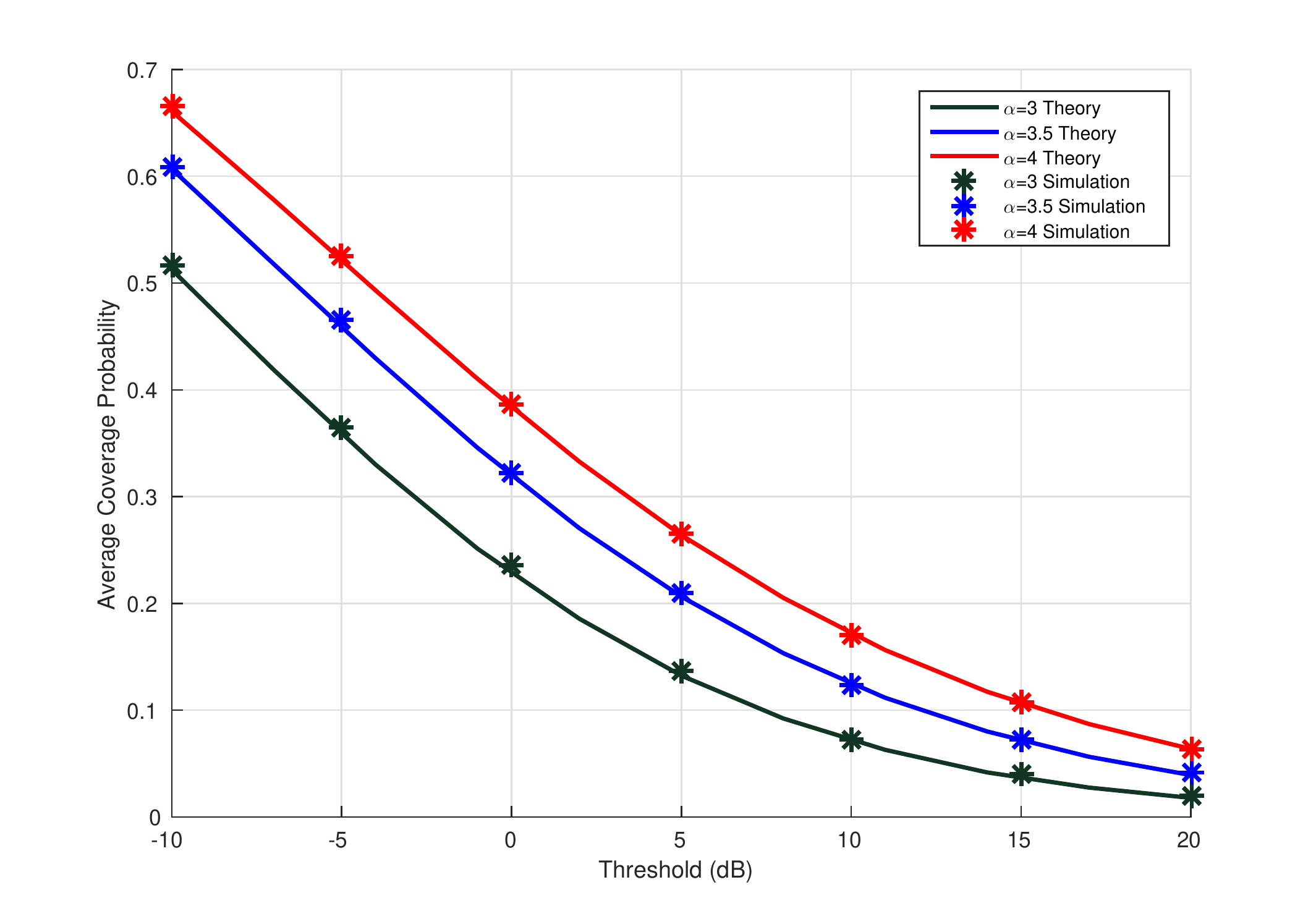}
\caption{Variation of coverage probability with threshold $T (dB)$ and different values of pathloss exponent $\alpha$}
\label{fig:singleSNR}
\end{figure}
Figure \ref{fig:singleSNR} indicates that when coverage threshold $T_c=0$ dB and $SNR=10 dB$, pathloss exponent $\alpha$ increases from 3.0 to 3.5 and ends at 4.0, the average coverage probability will increase by $36.66\%$ and $63.8\%$. The variance of average coverage probability with different values of $\alpha$ is shown in Table \ref{table:Tc0coverage}.
\begin{table}[!h]
\centering
\small
\begin{tabular}{|l|c|c|c|}
\hline \rule[0ex]{0pt}{2.5ex} Path loss exponent $\alpha$ & 3.0  & 3.5 & 4\\
\hline \rule[0ex]{0pt}{2.5ex} Average coverage probability  & 0.2362  & 0.3228 & 0.387 \\
\hline 
\end{tabular}
\caption{Average coverage probability when $T_c=0, SNR=10$}
\label{table:Tc0coverage}
\end{table}

When the coverage threshold increases that means the UE need a higher received SINR to detect and decode the received signals, the probability of successful communication between the user and its associated BS reduces which is reflected  in the decrease of coverage probability as shown in Figure  \ref{fig:singleSNR}. It is observed that when the coverage threshold increases  from 0 dB to 5 dB, the average coverage probability reduces by around $42.4\%$ from 0.2362 to 0.136.

\begin{figure}[!h]
\centering
\includegraphics[width=9cm,height=6.5cm]{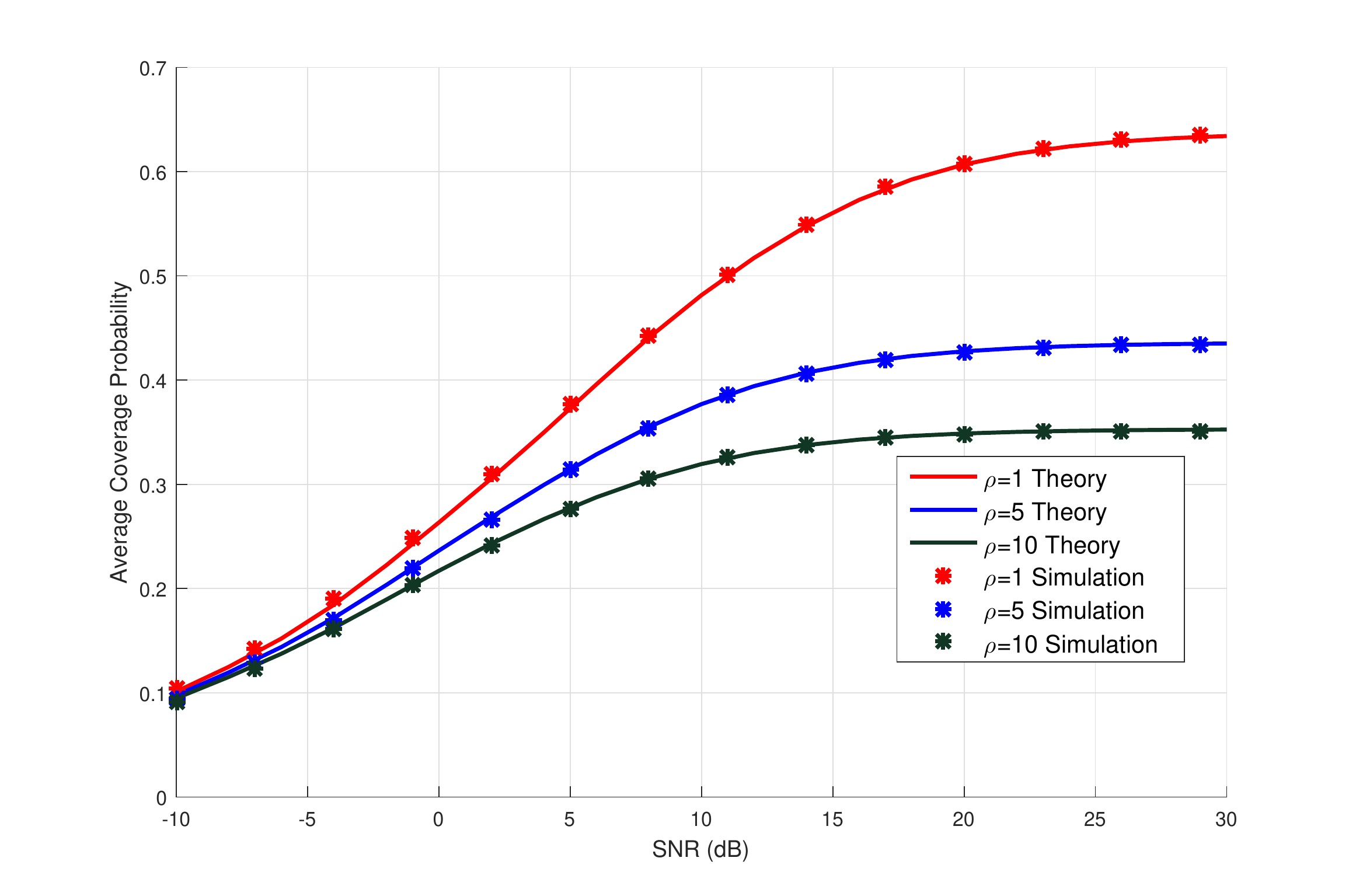}
\caption{Variation of average coverage probability with $SNR (dB)$}
\label{fig:multiSNR}
\end{figure}

When the transmission power P is much greater than the power of Gaussian noise, i.e. $P>>\sigma^2$, the Equation \ref{eq:receivedSINR} can be approximated by 
\begin{equation}
SINR(r)=\frac{\zeta}{\rho}\frac{gr^{-\alpha}}{\sum_{u\in\theta} \tau(RB_i=RB_j)g_ur_u^{-\alpha}}
\label{eq:aproxreceivedSINR}
\end{equation}
Hence in this case, the average coverage probability is consistent with the changes of standard transmission power $P$.
Figure \ref{fig:multiSNR} indicates that the average coverage probability is proportional to the standard transmission power when $SNR<20$ dB and reaches the upper bound when $SNR>20$ dB. Furthermore, it is observed that the upper bound is inversely proportional to the transmission power ratio. For example, when the transmission power ratio increase by 5 times from 1 to 5, the upper bound reduces by 30\% from 0.6 to around 0.42.

\begin{figure}[!h]
\centering
	\includegraphics[width=9cm,height=6.5cm]{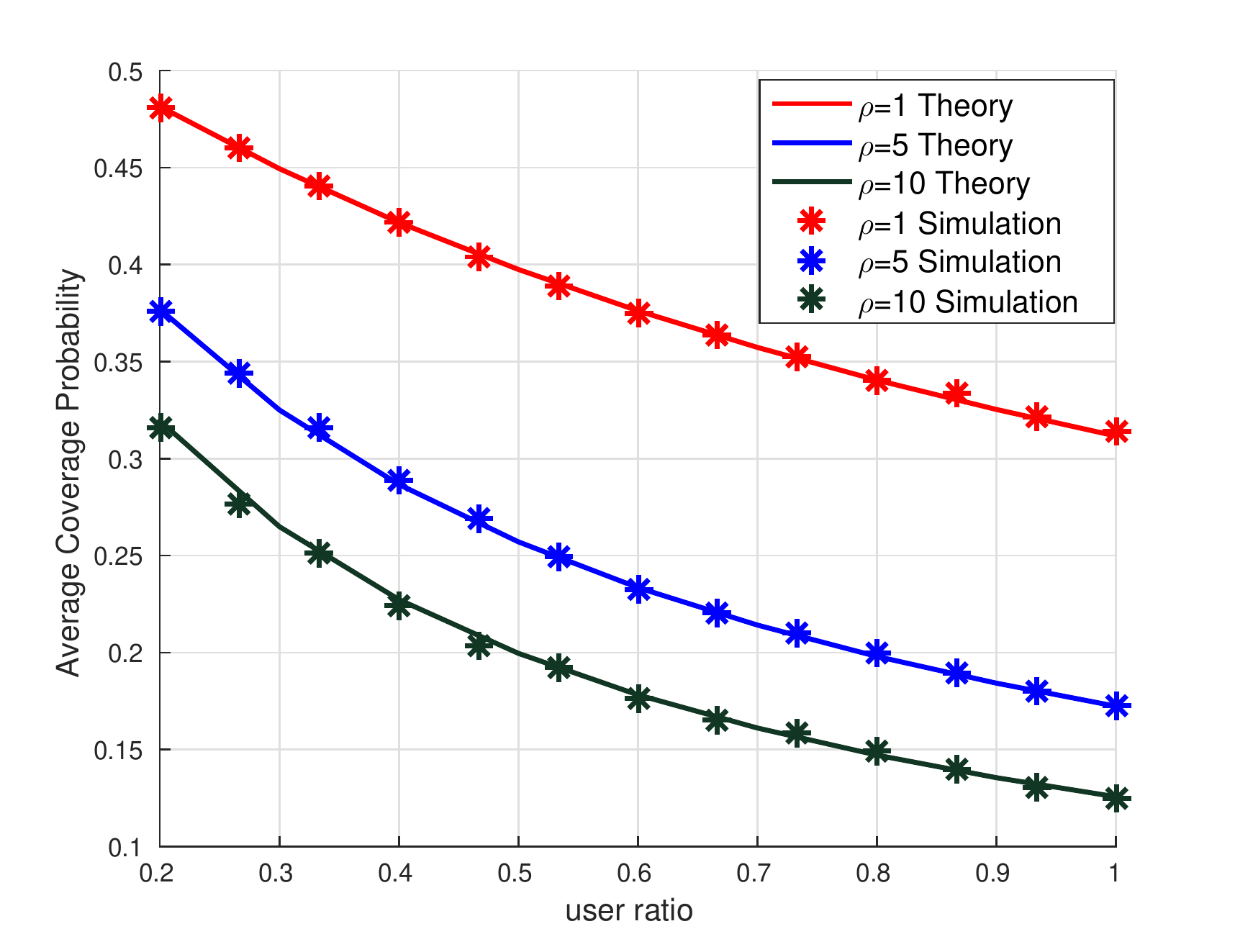}
	\caption{Average coverage probability with different values of user ratio $\epsilon$}
	\label{fig:differentuserratio}
\end{figure}

The impact of the ratio between the number of users and RBs (i.e. user ratio) is presented in Figure \ref{fig:differentuserratio}. When the user ratio increases, it means that more users  have  connections with the BS and more RBs should be used. Hence, the probability which two BSs transmit on the same RB at the same time increase which result in an increase  of the ICI. Consequently, the average coverage probability reduces.

\begin{figure}[!h]
\centering
	\includegraphics[width=9cm,height=6.5cm]{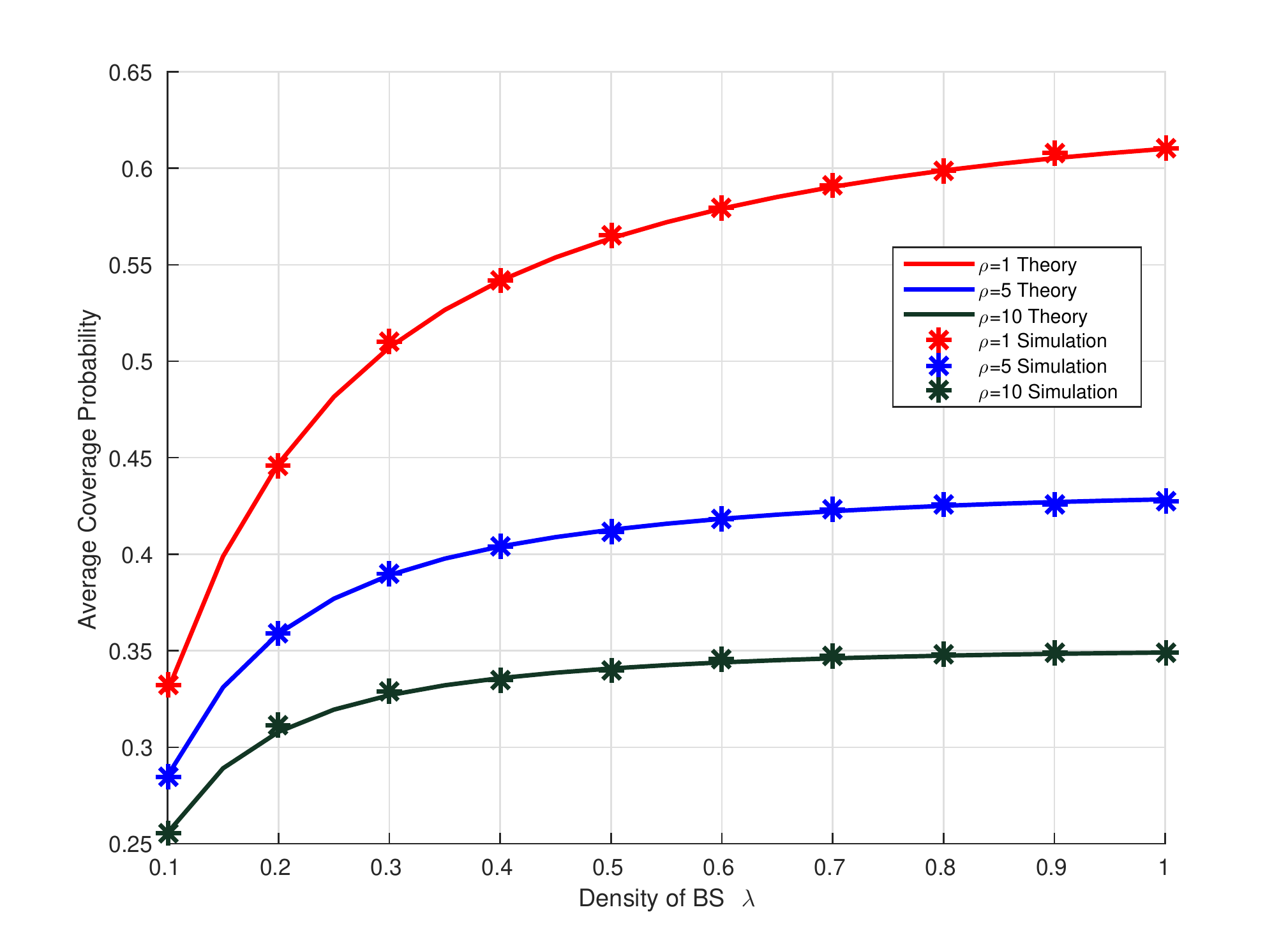}
	\caption{Average coverage probability with different values of density of base station $\lambda$}
	\label{fig:differentlambda}
\end{figure}
It is clear that an increase in the density of BSs $\lambda$ means that the user has more opportunities to connect with the BS and the distance from the users and its serving BS may be reduced. However, when the density of BSs increases, the number of interfering BSs increases. Hence, the power of the interfering BSs in this case is counter-balanced by the power of the serving BS. Consequently, average the coverage probability does not depend on the density of the BS as shown in Figure \ref{fig:differentlambda}.

\begin{figure}[!h]
\centering
	\includegraphics[width=9cm,height=6.5cm]{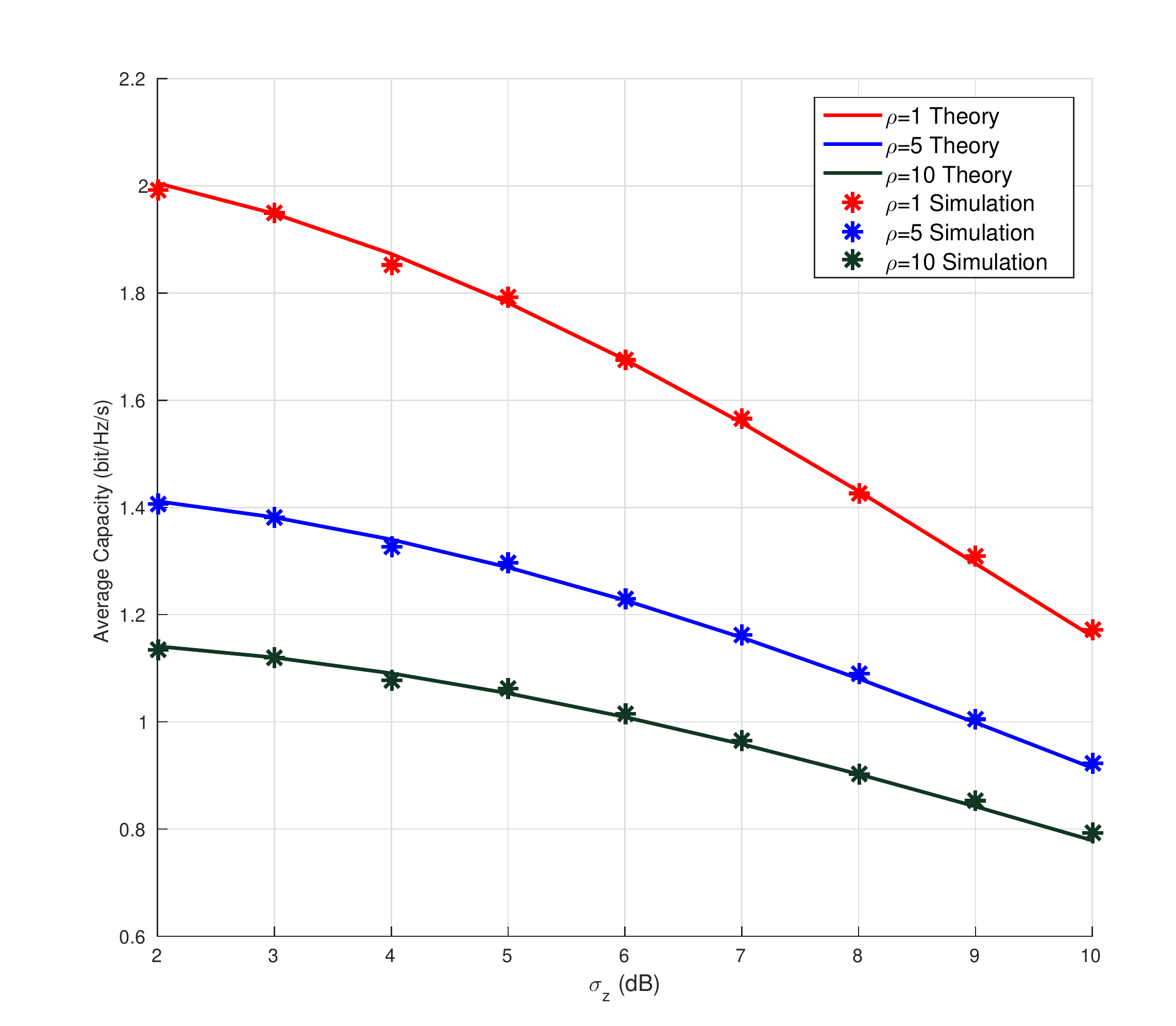}
	\caption{Variation of average capacity with $\sigma_z(dB)$}
	\label{fig:capacity}
\end{figure}

The square of the variance of Lognormal random variable $\sigma_z$, i.e. $\sigma_z^2$, denotes the power of the fading channel. That means if $\sigma_z$ increases, the signal will be more strongly affected by the fading. Hence, the average capacity is inversely proportional to the $\sigma_z$. Figure \ref{fig:capacity} indicates that when the power of fading channel  doubles  from 5 dB to 8 dB, the average data rate reduces by $20.42\%$ from 1.792 to 1.426 (bit/Hz/s) in the case of $\rho=1$, i.e. all BSs have the same transmission power. 

\begin{figure}[tbph]
\centering
\includegraphics[width=9cm,height=6.5cm]{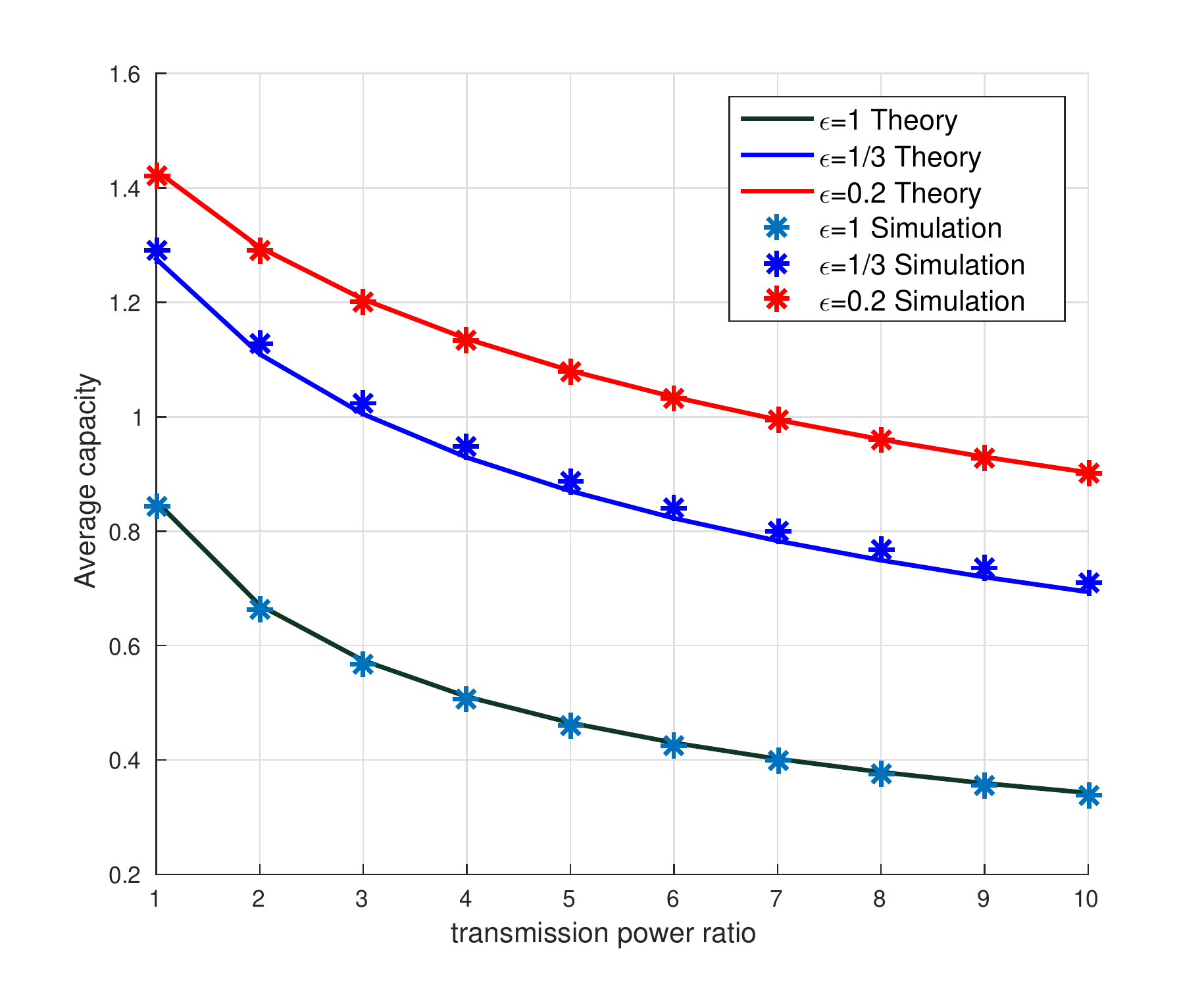}
\caption{Capacity with different values of transmission power ratio between the interfering and serving BSs.}
\label{fig:capacitypowerratio}
\end{figure}

In all simulation results, the power adjustment coefficient of the serving BS $\zeta$ is set to 1 while the coefficient of the interfering BS $\rho$ can take  three values 1, 5 and 10 from Figure \ref{fig:multiSNR} to \ref{fig:capacity} and from 1 to 10 in Figure \ref{fig:capacitypowerratio}. Hence, in this case $\rho$ represents the ratio between the interfering and serving BS of the typical user. The effects of power ratio on user's performance are demonstrated  through the gap between curves with different values of $\rho$ and highlighted in the Table \ref{performancedifferntrho}. 
\begin{table}[!h]
\centering
\small
\begin{tabular}{|l|c|c|c|}
\hline \rule[0ex]{0pt}{2.5ex} Power ratio  & 1  & 5 & 10\\
\hline \rule[0ex]{0pt}{2.5ex} Average coverage probability  & 0.4815  & 0.3770     & 0.3195 \\
							                                &         & (-21.70\%) & (-33.64\%)\\
\hline \rule[0ex]{0pt}{2.5ex} Average capacity				& 1.426   & 1.089      & 0.9037 \\
								                            &         & (-23.63\%)  & (-36.63)\\
\hline 
\end{tabular}
\caption{Performance of user with different values of $\rho$ (SNR = 10, user ratio = 0.2, $\lambda=0.25$)}
\label{performancedifferntrho}
\end{table}

In the Table \ref{performancedifferntrho}, the  negative percentage represents the percentage by which the user's performance, e.g. average coverage probability and average capacity, reduce when compared to those in the case when power ratio equals 1. For example, $-21.70\%$ and $-33.64\%$ mean the average coverage probability decreases by $21.70\%$ and $33.64\%$ when the power ratio increase from 1 to 5 and ends at 10.  

\subsection{The accuracy of simulation}
The accuracy of simulation is represented through the variance of the simulation results which is defined by 
\begin{equation}
var(X)=\frac{1}{NoS}{\sum\limits_{i=1}^{NoS}(x_i-\hat{x}_i)^2}
\label{eq:variance}
\end{equation}
in which 
\begin{itemize}
	\item NoS is the number of simulations
	\item $x_i$ is the simulation result at $i^{th}$ run.
	\item $\hat{x}_i$ is the expected vale of NoS simulation times.\\
	\begin{equation*}
	\hat{x}_i=\frac{1}{NoS}\sum\limits_{i=1}^{NoS}x_i
	\end{equation*}
\end{itemize}
In simulation, the results are obtained by taking the average values from 5 runs , the number of samples in each run is upto $10^5$ (sample). The variances of the  results obtained is shown in figures are presented in Figure \ref{fig:variance}.
\begin{figure}[!h]
\hspace{-1.3cm}
\includegraphics[width=11cm,height=8.5cm]{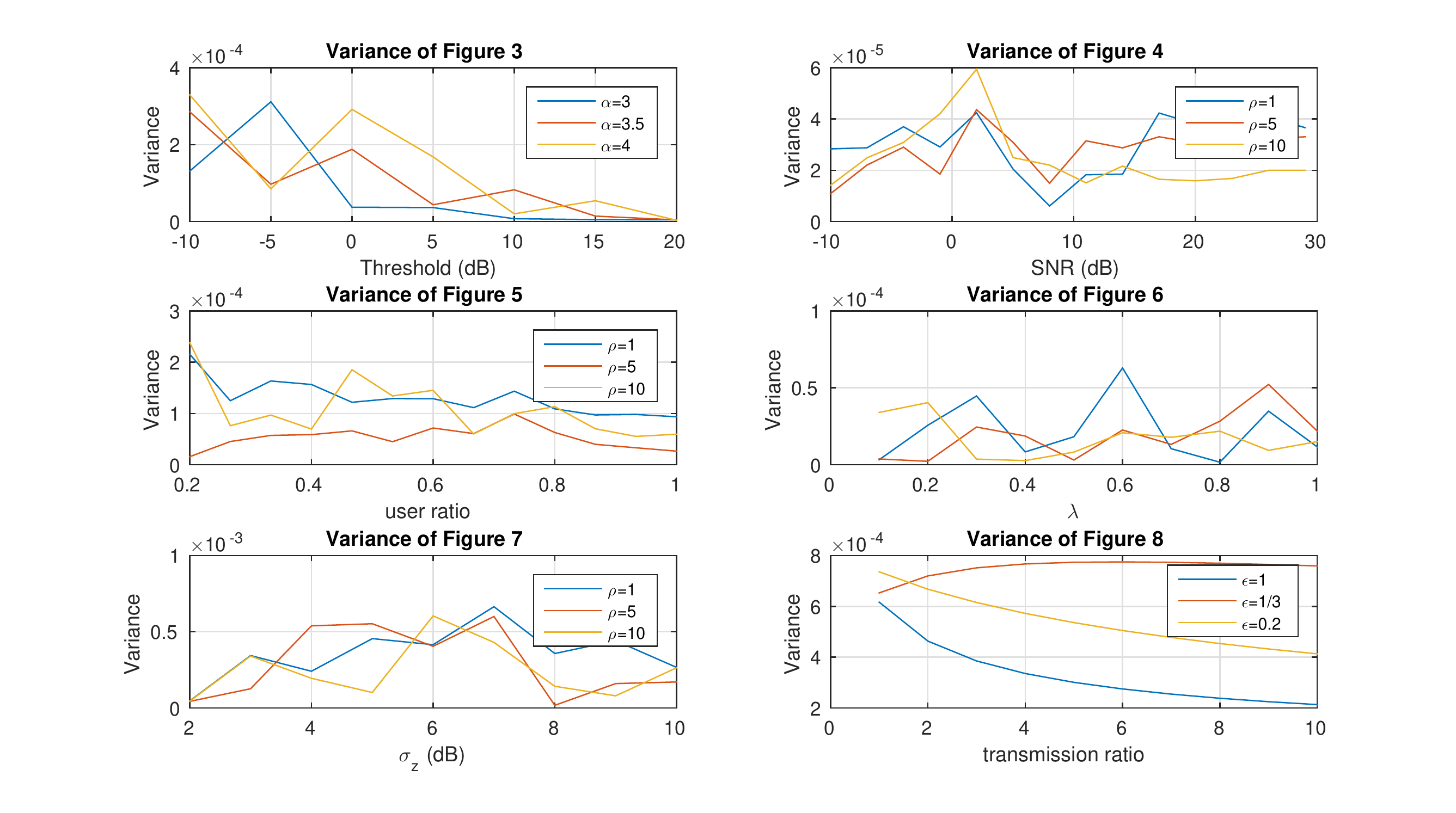}
\caption{Variance of simulation results}
\label{fig:variance}
\end{figure}

It is observed that in all cases, the variance of simulation results are smaller than $10^{-3}$. Hence, it is said that the  results obtained from simulation are accurate and stable.
\section{Conclusion}
In this paper, the performance of the typical user in terms of coverage probability and capacity in the PPP network in Rayleigh-Lognormal fading channel  was  presented. The analytical results for the network with $M$ users $N$ user in each cell are comparable with the corresponding published results for the network with either a user or a RB. Furthermore, the paper assumed that the interfering  and serving BSs have different transmission power. This assumption corresponds to the differences between the transmission power of BSs in different tiers or even in a given tier.
The numerical results show  that when the coverage threshold which represents the sensitivity of UE increased  three times from 0 to 5 dB, the average coverage probability reduces by around 42.2\%. Furthermore, when the power ratio between the transmission power of interfering and serving BS increased from 1 to 5 and ends at, the average capacity of a link reduced by 23.63\% and 36.63\% , respectively.
\section{APPENDIX}
The coverage probability of a typical user, which is located in cell $i$ and served on RB $\delta$, is defined in Equation \ref{coveragadefinition}:
\begin{align*}
&\mathbb{P}(SINR(r)>T_c)=\mathbb{P}\left(\frac{\zeta Pgr^{-\alpha}}{I_\theta+\sigma^2} > T_c \right) \nonumber \\
&  =\mathbb{P}\left(g < \frac{T_cr^\alpha(I_\theta+\sigma^2)}{\zeta P} \right) \nonumber \\
& = \mathbb{E}\left[\sum_{n=1}^{N_p} \frac{\omega_n}{\sqrt{\pi}}\exp\left(-\frac{T_cr^\alpha(I_\theta+\sigma^2)}{\zeta P\gamma(a_n)} \right) \right] \nonumber  \\
& = \sum_{n=1}^{N_p}\frac{w_n}{\sqrt{\pi}} \mathbb{E}\left[\exp\left(-\frac{T_cr^\alpha(I_\theta+\sigma^2)}{\zeta P\gamma(a_n)} \right)  \right]  \nonumber \\
& = \sum_{n=1}^{N_p}\frac{w_n}{\sqrt{\pi}}\exp\left(-\frac{T_cr^\alpha\sigma^2}{\zeta P\gamma(a_n)} \right) \mathbb{E}\left[\exp\left(-\frac{T_cr^\alpha I_\theta}{\zeta P\gamma(a_n)} \right)  \right] \nonumber  \\
&= \sum_{n=1}^{N_p}\frac{\omega_n}{\sqrt{\pi}}\exp\left(-\frac{T_cr^\alpha}{\gamma(a_n)}\frac{1}{\zeta SNR} \right)\mathbb{E}\left(\exp\left(-f(n)I_\theta \right) \right)
\end{align*}
in which $\frac{T_cr^\alpha}{\zeta P\gamma(a_n)}=f(n)$. $SNR$ is the standard transmission power-noise ratio at the base station. Considering the expectation and given that the ICI was defined in Equation \ref{eq:interfernce}
\begin{align*}
 &= \mathbb{E}\left[\exp\left(-f(n)\sum_{u \in \theta}\tau(RB_i=RB_j)\rho Pg_ur_u^{-\alpha} \right)\right] \nonumber\\
& = \mathbb{E}_\theta\left[\mathbb{E}_{g_u}\prod_{u\in\theta}\tau(RB_i=RB_j)\exp\left(- f(n)\rho Pg_ur_u^{-\alpha} \right)  \right] \nonumber \\
& = \mathbb{E}_\theta\left[\prod_{u\in\theta}\mathbb{E}_{g_u}\epsilon \exp\left(- f(n)\rho Pr_u^{-\alpha}g_u \right)  \right]
\end{align*}
Since $g_u$ is Rayleigh-Lognormal fading channel whose MGF is calculated from Equation \ref{eq:MGF}, then
\begin{align}
& = \mathbb{E}_\theta\left[\prod_{u\in\theta}\sum_{n_1=1}^{N_p}\frac{\omega_{n_1}}{\sqrt{\pi}} \frac{\epsilon}{1+\gamma(a_{n_1})f(n)\rho Pr_u^{-\alpha}}\right]
\label{eq:coverdefi}
\end{align}

Using the properties of PPP probability generating function \cite{Stegun1972}
\begin{equation}
=\exp\left(-\frac{\pi\lambda }{N}\left(\sum_{n_1=1}^{N_p}\frac{\omega_{n_1}}{\sqrt{\pi}}\int_{r}^{\infty}2\frac{\gamma(a_{n_1})f(n)\rho Pr_u^{-\alpha}}{1+\gamma(a_{n_1})f(n)\rho Pr_u^{-\alpha}}  \right)r_udr_u  \right) 
\label{expeq}
\end{equation}
Given   that $\frac{T_cr^\alpha}{\zeta P\gamma(a_n)}=f(n)$ and letting $t=\left(\frac{r_u }{r} \right)^2$, $C=T_c\frac{\rho}{\zeta}\frac{\gamma(a_{n_1})}{\gamma(a_n)}$ then the integral becomes 
\begin{align*}
= & r^2\int_{1}^{\infty}\left(1-\frac{1}{1+Ct^{-\alpha/2}} \right)dt \\
= & r^2 \left[\int_{0}^{\infty} \frac{Ct^{-\alpha/2}}{1+Ct^{-\alpha/2}}dt-\int_{0}^{1} \frac{Ct^{-\alpha/2}}{1+Ct^{-\alpha/2}}dt \right] \\
= & r^2 (I_1-I_2)
\end{align*}
Using properties of Gamma function \cite{Stegun1972}, the first integral $I_1$ is obtained by
\begin{align}
I_1=\frac{2}{\alpha}C^{\frac{2}{\alpha}}\frac{\pi}{\sin\left(\frac{\pi(\alpha-2)}{\alpha} \right) }
\label{eq:I1}
\intertext{The second integral $I_2$ is approximated by using Gauss-Legendre rule \cite{Stegun1972}}
I_2=\sum_{m=1}^{N_{GL}}\frac{c_m}{2}\frac{C}{C+\left(\frac{x_m+1}{2} \right)^{\alpha/2}}
\label{eq:I2}
\end{align}
For accurate computation, $N_{GL}=10$ is chosen.
Subsequently, the expectation can be approximated by
\begin{dmath}
 =\exp\left[-\pi\lambda\epsilon r^2 \left( \sum_{n_1=1}^{N_p}\frac{\omega_{n_1}}{\sqrt{\pi}}\left(\frac{2}{\alpha}C^{\frac{2}{\alpha}}\frac{\pi}{\sin\left(\frac{\pi(\alpha-2)}{\alpha} \right) } \\ \qquad+  \sum_{m=1}^{N_{GL}}\frac{c_m}{2}\frac{C}{C+\left(\frac{x_m+1}{2} \right)^{\alpha/2}} \right) \right)\right]  \nonumber \\
 = \exp\left(-\pi\lambda\epsilon r^2 f_{I}(T_c,n) \right) 
\label{eq:singlefinalexp}
\end{dmath}
Substituting Equation \ref{eq:coverdefi} - \ref{eq:singlefinalexp}, the Theorem \ref{theo:singlecoverage} is proved.

\small{
\addcontentsline{toc}{chapter}{Bibliography}
\bibliographystyle{IEEEtran}
\bibliography{refDA}
}
\end{document}